\documentclass[12pt]{amsart}
\usepackage[latin1]{inputenc}
\usepackage{tikz}
\usetikzlibrary{arrows.meta}
\usetikzlibrary{bending}
\usepackage{verbatim}
\usepackage{xcolor}
\usepackage[all,color]{xy}
\usepackage{dblfloatfix,afterpage,float,longtable,tensor,mathrsfs,import,stmaryrd,wasysym,amsmath,amssymb,amsthm,amscd,amsxtra,amsfonts,mathrsfs,graphicx,enumerate,bm,slashed,array,setspace,enumitem,minibox,booktabs,rotating,multirow,adjustbox,hhline,arydshln,placeins}
\usepackage{tikz-feynman}
\input xy
\xyoption{all}
\usetikzlibrary{shapes,arrows}
\usetikzlibrary{positioning}
\usetikzlibrary{calc}
\input xy
\xyoption{all}
\newtheorem{Theorem}{Theorem}[section]

\newtheorem{Proposition}[Theorem]{Proposition}
\newtheorem{Corollary}[Theorem]{Corollary}

\newtheorem{Remark}[Theorem]{Remark}

\newtheorem{Definition}[Theorem]{Definition}

\newtheorem*{Definition*}{Definition A}
\newtheorem*{Theorem*}{Theorem}
\DeclareMathOperator{\im}{im}
\DeclareMathOperator{\vol}{vol}

\newcommand*{\overbar}[1]{\mkern 1.5mu\overline{\mkern-1.5mu#1\mkern-1.5mu}\mkern 1.5mu}
\advance\evensidemargin-.5in
\advance\oddsidemargin-.5in
\advance\textwidth1in

\definecolor{darkgreen}{rgb}{0.0, 0.5, 0.0}
\definecolor{darkred}{rgb}{0.7, 0.11, 0.11}

\begin{document}

\author{Charlie Beil}
\address{Institut f\"ur Mathematik und Wissenschaftliches Rechnen, Universit\"at Graz, Heinrichstrasse 36, 8010 Graz, Austria.}
\email{charles.beil@uni-graz.at} 
\title[A derivation of the first generation masses from internal spacetime]{A derivation of the first generation particle masses from internal spacetime} 

\begin{abstract}
Internal spacetime geometry was recently introduced to model certain quantum phenomena using spacetime metrics that are degenerate. We use the Ricci tensors of these metrics to derive a ratio of the bare up and down quark masses, obtaining $m_u/m_d = 9604/19683 \approx .4879$. This value is within the lattice QCD value $.473 \pm .023$, obtained at $2 \operatorname{GeV}$ in the minimal subtraction scheme using supercomputers. Moreover, using the Levi-Cevita Poisson equation, we derive ratios of the dressed electron mass and bare quark masses. For a dressed electron mass of $.511 \operatorname{MeV}$, these ratios yield the bare quark masses $m_u \approx 2.2440 \operatorname{MeV}$ and $m_d \approx 4.599 \operatorname{MeV}$, which are within/near the lattice QCD values $m^{\overline{\operatorname{MS}}}_u = (2.20\pm .10) \operatorname{MeV}$ and $m^{\overline{\operatorname{MS}}}_d = (4.69 \pm .07) \operatorname{MeV}$. Finally, using $4$-accelerations, we derive the ratio $\tilde{m}_u/\tilde{m}_d = 48/49 \approx .98$ of the constituent up and down quark masses. This value is within the $.97 \sim 1$ range of constituent quark models. All of the ratios we obtain are from first principles alone, with no free or ad hoc parameters. Furthermore, and rather curiously, our derivations do not use quantum field theory, but only tools from general relativity. 
\end{abstract}

\maketitle

\tableofcontents

\section{Introduction} \label{Introduction}

Internal spacetime is a modification of general relativity in which spacetime metrics are degenerate along the worldlines of fundemantal particles. 
The aim of this modification is to resolve the measurement problem for certain quantum systems using spacetime geometry. 
For example, tangent spaces on an internal spacetime have variable dimension, and the projection from one tangent space to another corresponds to spin wavefunction collapse. 
In particular, tangent space projection is able to reproduce the Born rule for both spin and polarization. 

We briefly define the geometry. 
Let $(\tilde{M},\tilde{g})$ be a Lorentzian manifold with metric $\tilde{g}_{ab}: \tilde{M}_p \otimes \tilde{M}_p \to \mathbb{R}$, where $\tilde{M}_p$ is the tangent space at a point $p \in \tilde{M}$.
We endow $\tilde{M}$ with a (locally finite) collection of fundamental particles called \textit{pointons}.
In this article we will only consider pointons with timelike worldlines.
If a pointon worldline $\beta \subset \tilde{M}$ intersects $p \in \tilde{M}$, then we define the \textit{internal metric} 
\begin{equation} \label{h}
h_{ab}: \tilde{M}_p \otimes \tilde{M}_p \to \mathbb{R}
\end{equation}
to be the projection $h_{ab} = g_{ab} + v_av_b$ of the tangent $4$-vector $v^a$ (i.e., $4$-velocity) to $\beta$ at $p$. 
In general, $h_{ab}$ projects out the tangent vector to each pointon worldline intersecting $p$. 
We call the subspace image $M_p := \im h \subset \tilde{M}_p$ the \textit{internal tangent space} at $p$, and the resulting new geometry \textit{internal spacetime},
\begin{equation*}
M := (\tilde{M} \setminus (\cup_i \beta_i)) \cup (\cup_j \{ \beta_j \}).
\end{equation*}
Consequently, time is stationary along the worldlines of pointons, that is, \textit{pointon worldlines are single points of spacetime $M$}.  

To see how this geometry may be useful in modeling certain quantum phenomena, consider two particles that are entangled at a point $p \in \tilde{M}$ and are then separated. 
If time is stationary along their worldlines, then regardless of how far apart the particles are in space, they remain `touching' at $p$ in spacetime---\textit{their correlation is nonlocal in space but local in spacetime.} 

There is an immediate problem with the model, however: it is based on point particles that follow geodesics in spacetime, and such particles will generically never meet.
Consequently, pointons should almost never interact.
This problem does not arise in a (classical or quantum) field theory because a field excitation is spread out over a small region of space, peaked at the point of excitation.
Nevertheless, if we modify $h_{ab}$ in (\ref{h}) so that it becomes continuous, then the problem disappears.
Specifically, in Theorem \ref{pointon stress-energy theorem} we show that a single stationary pointon is equivalent to a fluid with pressure and shear viscosity, that is, \textit{a pointon turns the spacetime around it into a honey-like fluid.}
Consequently, although pointons continue to follow geodesic trajectories, they behave as field excitations spread out over small regions of space.

An internal spacetime model of the standard model particles, derived from the free Dirac Lagrangian, was recently introduced in \cite{B1,B2,B4} and is outlined in Section \ref{Preliminary section}.
In this article we investigate the geometry of the first generation fermions (the electron neutrino $\nu_e$, electron $e$, up quark $u$, and down quark $d$).
The simplest degenerate geometric configurations that are able to represent these particles are:
\begin{center}
\begin{tabular}{rcl}
particle && degenerate worldvolume in $M$\\
\hline
$\nu_e$ & $\longleftrightarrow$ & $\varnothing$ (a free spinor) \\
$e$ & $\longleftrightarrow$ & a point (a \textit{pointon})\\
$u$ & $\longleftrightarrow$ & a sphere (a \textit{pointal sphere})\\
$d$ & $\longleftrightarrow$ & a union of a point and a sphere (a \textit{pointal union})
\end{tabular}
\end{center}
In Section \ref{first section} we modify the internal metrics $h_{ab}$ of these four configurations in Mink-owski space $\tilde{M} = \mathbb{R}^{1,3}$ so that they (i) are continuous, and (ii) support a spinor.
In particular, the spinor of an on-shell pointon follows a circular null geodesic. 
We denote these new metrics $g_{ab}$.
 
With the four metrics in hand, we use the Einstein field equation and generalized Poisson equation $\vec{\nabla}^2 \sqrt{-g_{tt}} = -R^t_t \sqrt{-g_{tt}}$ in Sections \ref{energy density section}--\ref{gravitational mass section} to determine their energy densities and masses.
Recall that a quark's constituent mass is the amount of binding energy required to make a hadron emit a meson containing the quark. 
It is an effective mass that includes contributions from gluons and sea-quarks, and is hadron-dependent.
Denote by $m_e$ (resp.\ $\tilde{m}_e$) the bare (resp.\ dressed) electron mass, and by $m_{u/d}$ (resp.\ $\tilde{m}_{u/d}$) the bare (resp.\ constituent) up/down quark masses.
Our main result is the following (see Theorems \ref{sphere metric theorem2}, \ref{energy densities}, \ref{constituent}, and Corollaries \ref{main corollary1}, \ref{positive mass}).

\begin{Theorem*}
On an asymptotically flat internal spacetime, the Einstein field equation and generalized Poisson equation yield
\begin{equation*}
m_e = 0, \ \ \ \ \tilde{m}_u/\tilde{m}_d = 48/49 \approx .9796,
\end{equation*}
\begin{equation} \label{three ratios}
\tilde{m}_e/m_u = 2187/9604, \ \ \ \ \tilde{m}_e/m_d = 1/9, \ \ \ \ m_u/m_d = 9604/19683 \approx .4879.
\end{equation}
In contrast to the other masses, the neutrino mass $m_{\nu_e}$ is not fixed.
Furthermore, the (gravitational) mass of each configuration is positive.
\end{Theorem*}

In the standard model (with the Higgs boson), the bare electron mass is zero, $m_e = 0$, and thus agrees with our model. 
If we take the dressed electron mass to be $\tilde{m}_e = .511 \operatorname{MeV}$, then the ratios (\ref{three ratios}) imply the bare quark masses
\begin{equation*}
m_u \approx 2.2440 \operatorname{MeV},
\ \ \ \ \ \ \ \
m_d \approx 4.599 \operatorname{MeV}.
\end{equation*}
These are within/near the average lattice QCD values in the minimal subtraction ($\overbar{\text{MS}}$) scheme at $\mu = 2 \operatorname{GeV}$, with four active light quark flavors.
Recent estimates given by the Particle Data Group \cite[60.7--60.9]{PDG} are
\begin{equation*}
\begin{split}
\overline{m}_u & = (2.20\pm .04 \pm .06[\pm .07]) \operatorname{MeV},\\
\overline{m}_d & = (4.69 \pm .03 \pm .04[\pm .05]) \operatorname{MeV},\\
\overline{m}_u/\overline{m}_d & = .473 \pm .007 \pm .016[\pm .017]. 
\end{split}
\end{equation*}
The first error is statistical, the second is systematic, and the value in brackets is the square root of the sum of the squared errors.
The estimates are obtained by averaging over models where vacuum polarization effects result from up, down, and strange sea-quarks ($N_f = 2 +1$), and up, down, strange, and charm sea-quarks ($N_f = 2+1+1$); in both cases the up and down quarks are assumed to be degenerate \cite[Sections 60.4, 60.5.1]{PDG}.
(To note, these estimates differ from the PDG particle listings since they are based solely on lattice computations, whereas in the particle listings non-lattice estimates are included.
In addition, both the lattice estimates that are retained and the averaging procedures differ \cite[Section 60.4]{PDG}.)

Furthermore, the constituent quark mass ratio we obtain, namely $\tilde{m}_u/\tilde{m}_d = 48/49$, is within the range obtained from different constituent quark models, $.97 \sim 1$ (see \cite[p 630]{MW}, \cite[p 8]{OSY}, \cite[IV.B.4]{KR}, and references therein).
Our value is determined from the $4$-accelerations $\nabla_{e_t}e_t$ of the corresponding internal spacetime metrics.

Finally, in Section \ref{neutrino section} we graph the energy densities of the pointal configurations.
We find that, away from their `interaction radii', up and down quarks generate positive energy densities while electron neutrinos generate negative energy densities---even though they have positive gravitational mass.
We speculate that this surprising property of electron neutrinos could account for dark matter, since gravitation also yields a negative energy density.

Degenerate spacetime metrics have appeared in other contexts in the literature, such as in the regularization of the big bang singularity \cite{Ba,K1,K2,K3,KW} and loop quantum gravity \cite{LW}.
Recently, \cite{GB} established a unique covariant pseudoinverse to degenerate metrics of constant rank, thus enabling a novel generalization of general relativity to metrics which are degenerate. We note, though, that the metrics we consider here do not have constant rank.
Our use of metric degeneracy originates from nonnoetherian algebraic geometry, which in turn arose from the vacuum geometry of certain quiver gauge theories in string theory (see \cite{F-K,B3} and references therein).

\textbf{Notation:} Tensors labeled with upper and lower indices $a,b, \ldots$ represent covector and vector slots respectively in Penrose's abstract index notation (so $v^a \in V$ and $v_a \in V^*$), and tensors labeled with indices $\mu, \nu, \ldots$ denote components with respect to a coordinate basis. 
Throughout, we will use natural units $\hbar = c = G = 1$ and the signature $(-,+,+,+)$.
We denote by $\kappa := 8 \pi G$ the normalized gravitational constant, and by $d\Omega^2 := d \theta^2 + \sin^2 \theta d \phi^2$ the spherical metric.

\section{Standard model particles from internal spacetime geometry} \label{Preliminary section}

In this section we outline the construction of the standard model particles from internal spacetime, introduced in \cite{B2,B4}; the specific model of color charge given in (\textsc{b}) below is new.
For brevity, set $h := h_{ab}$ in (\ref{h}).

Let $\beta \subset \tilde{M}$ be the worldline of a pointon; then $\beta$ is a single point of spacetime $M$.
In particular, the pointon's $4$-velocity $v$ vanishes,
\begin{equation*}
v \in \ker h.
\end{equation*}
We must therefore replace $v$ with something that lives in the internal tangent space $M_{\beta(t)} := \im h \subseteq v^{\perp} \subset \tilde{M}_{\beta(t)}$, where $v^{\perp}$ is the orthogonal subspace to $v$ in $\tilde{M}_{\beta(t)}$. 

Recall that a unit vector $\vec{v}$ in $\mathbb{R}^3$ is specified by the plane $\vec{v} \cdot \vec{x} = 0$ orthogonal to $\vec{v}$, together with a choice of orientation.
Indeed, if $e_1 = \vec{v}, e_2, e_3$ is an orthonormal basis, then $\vec{v} = \pm e_2 \times e_3$.
This is generalized to higher dimensions using differential forms: the orthogonal subspace of a (co)vector $v$ is specified by its Hodge dual $\star v$.
In our example, $\star \vec{v} = o e_2 \wedge e_3$, where $o \in \{ \pm 1 \}$ is a choice of orientation.
Just as an orientation is needed define the cross product in $\mathbb{R}^3$, an orientation is needed to define a Hodge dual, that is, Hodge duals are \textit{pseudo-forms}.

We may thus replace the $4$-velocity $v$ of an isolated pointon with its Hodge dual $\star v$, together with a timelike orientation $o_0 \in \{ \pm 1\}$.
Since a choice of orientation $o_{\tilde{M}} \in \{ \pm 1 \}$ of $\tilde{M}_{\beta(t)}$ is unphysical, this new timelike orientation $o_0$ must be independent of $o_{\tilde{M}}$.
\textit{We identify $o_0$ with the electric charge of the pointon.}

\begin{Definition} \cite[3.1]{B4} \rm{
The \textit{internal $4$-velocity} of a pointon at $p = \beta(t) \in \tilde{M}$ with $4$-velocity $v \in \tilde{M}_p$ is the pseudo-form
\begin{equation*} \label{internal form}
\breve{v}_{a \cdots b} := o_{\operatorname{ker}h} \star \! \vol (\ker h) \in {\bigwedge} \! ^{\dim M_p} \, M_p^*,
\end{equation*}
where $\vol(\ker h)$ is the volume form for the kernel of $h$ at $p$, and $o_{\operatorname{ker}h} \in \{ \pm 1 \}$ is a free choice of orientation, independent of $o_{\tilde{M}}$.
}\end{Definition}

For each vanishing subspace of $\tilde{M}_{\beta(t)}$ we obtain a new free choice of orientation, and we identify these orientations with charge or spin,
\begin{equation*}
\textit{vanishing subspace} \ \ \ \longleftrightarrow \ \ \ \textit{choice of orientation} \ \ \ \longleftrightarrow \ \ \ \textit{charge or spin}
\end{equation*}
These identifications are given in Table \ref{c s table}.
For the remainder of this section, let $i,j,k \in \{ 1,2,3\}$ be distinct, and let $e_0,e_1, e_2, e_3$ be an orthonormal tetrad parallel transported along $\beta$.

\begin{table}
\caption{Charge and spin from vanishing subspaces of spacetime tangent spaces \cite{B2,B4}.
Here we assume that the unphysical choice of orientation $o_{\tilde{M}} = o_{0123} = o_0o_1o_2o_3 \in \{ \pm 1 \}$ of $\tilde{M}_{\beta(t)}$ is $1$.}
\label{c s table}
\begin{center}
\begin{tabular}{l:l:l}
\hline
vanishing subspace \ \ \ & free orientation \ \ \ & identification\\
\hline \hline
$e_0$ & $o_0 = o_{123}$ & electric charge $e^{o_0}$\\
\hline
$e_0 \wedge e_2 \wedge e_3$ & $o_1 = o_{023}$ & color charge $r^{o_1}$\\
$e_0 \wedge e_3 \wedge e_1$ & $o_2 = o_{013}$ & color charge $g^{o_2}$\\
$e_0 \wedge e_1 \wedge e_2$ & $o_3 = o_{012}$ & color charge $b^{o_3}$\\
\hline
$e_2 \wedge e_3$ & $o_{23} = o_{01}$ & spin in the direction $o_{01}e_1$\\
$e_3 \wedge e_1$ & $o_{13} = o_{02}$ & spin in the direction $o_{02}e_2$\\
$e_1 \wedge e_2$ & $o_{12} = o_{03}$ & spin in the direction $o_{03}e_3$\\
\hline
\end{tabular}
\end{center}
\end{table}

(\textsc{a}) \textit{Spin $\tfrac 12$ from $o_{jk}$:}
Suppose $\dim M_{\beta(t)} > 1$ for $t \in (t_0,t_1)$, and $\dim M_{\beta(t_0)} = 1 = \dim M_{\beta(t_1)}$. 
Then at the points $\beta(t_0), \beta(t_1) \in \tilde{M}$, the internal $4$-velocity is an actual (spacelike unit) vector,
\begin{equation*}
\breve{v} = o_{0jk}e_i = o_0o_{jk}e_i.
\end{equation*}
We identify the free orientation $o_{jk}$ with the pointon's spin, up or down, in the $e_i$ direction at $\beta(t_0)$ \cite[Section 4]{B3}.
Let $s$ be a $4$-vector parallel transported along $\beta$, such that $s$ projects onto $\breve{v}$ at $t_0$, that is, $s$ is a section of the projection at $t_0$.
The probability that a particular section $s$ is chosen is given by the \textit{Kochen-Specker probability} $p(s|\breve{v}) = \pi^{-1} s^a \breve{v}_a$.
Then, at $t_1$, $s$ is projected onto the line $M_{\beta(t_1)}$, becoming the internal $4$-velocity $\breve{v}$ there.
This tangent space projection corresponds to a measurement of spin in the direction $M_{\beta(t_1)}$ and \textit{exactly reproduces the Born rule} \cite{KS}. 
We call $s$ the \textit{spin vector} of the pointon.

The spin vector $s$ is realized geometrically as the angular momentum vector of an associated \textit{spinor particle} of angular velocity $\omega$ and radius $r_0$.
The particle's trajectory $\alpha \subset \tilde{M}$ defines a spinor $\psi \in \mathbb{C}^4$ by \cite[Definition 3.2]{B2}\footnote{Our model is fundamentally built from maximal covariantly independent sets of $\gamma^5$ and $\gamma^0$ eigenspinors derived in \cite[Propositions 3.5, 5.7, Lemma 5.8]{B2}, namely
\begin{equation*}
[0a] := \left[ \begin{smallmatrix} 0 \\ 0 \\ -e^{-ai \omega t}\\ 0 \end{smallmatrix} \right]_{\! \mathscr{C}},  \ \ 
[a0] := \left[ \begin{smallmatrix} e^{ai \omega t}\\ 0 \\ 0 \\ 0 \end{smallmatrix} \right]_{\! \mathscr{C}}, \ \ 
[ab] := [a0] + [0b] = \left[ \begin{smallmatrix} 0 \\ 0 \\ -e^{ai \omega t}\\ 0 \end{smallmatrix} \right]_{\! \mathscr{D}}, \ \ 
[\downarrow \! * ]/[* \! \downarrow ] := \left[ \begin{smallmatrix} e^{\pm i \omega t}\\ 0 \\ 0 \\ 0 \end{smallmatrix} \right]_{\! \mathscr{D}},
\end{equation*}
where $a,b \in \{ \uparrow \, = \! +1, \downarrow \, = \! -1\}$ with $a \not = b$; and $\mathscr{C}$, $\mathscr{D}$ denote chiral and Dirac bases.  
On an internal spacetime, $\gamma^5$ eigenspinors have electric charge and (color-neutral) $\gamma^0$ eigenspinors are neutral.}
\begin{equation*}
\psi = \gamma(h_{ab}\dot{\alpha}^b) = h_{ab}\gamma^a \dot{\alpha}^b \in \mathbb{C}^4.
\end{equation*}
Consequently, the pointon's rest energy $E_0$ and mass $m$ are also defined from $\alpha$,
\begin{equation} \label{E0 =}
E_0 = \omega = \hbar \omega, \ \ \ \ \ m = r_0^{-1} = \hbar/(c r_0),
\end{equation}
where units have been restored on the right.
Denoting the particle's tangential speed by $u = \omega r_0$, these relations imply \begin{equation*}
E_0 = mcu.
\end{equation*}
The pointon is then on shell if $u = c$, and off shell otherwise \cite[Section 4]{B2}.
From this characterization, off-shell particles (i) never violate relativity in that $k^2 = E_0^2$ always holds; and (ii) are \textit{classical equations of motion} of the modified Dirac Lagrangian,
\begin{equation} \label{modified Dirac}
\mathcal{L} = \bar{\psi}(i \slashed \partial - \omega) \psi = \bar{\psi}(i \slashed \partial - mu) \psi.
\end{equation}
We call $r_0$ the \textit{interaction radius} of the pointon. 

\begin{table} \label{first generation table}
\caption{The first generation fermions, $\nu_e, e, d, u$, are the geoms with exactly one nonzero spinor. Here, the indices $i,j,k \in \{ 1,2,3\}$ are distinct, and we have taken $o_{\tilde{M}} = 1$.}
\begin{center}
$\begin{array}{lll:l||lll:l||lll:l||lll:l}
\multicolumn{4}{l||}{\text{free spinor:}} & \multicolumn{4}{l||}{\text{pointon:}} & \multicolumn{4}{l||}{\text{pointal sphere:}} & \multicolumn{4}{l}{\text{pointal union:}}\\
o_0 & o_i & o_{jk} & \nu_e &
o_0 & o_i & o_{jk} & e &
o_0 & o_i & o_{jk} & u &
o_0 & o_i & o_{jk} & d\\ 
\hline
\varnothing & \varnothing & \textcolor{blue}{+} & [* \! \downarrow] &
\textcolor{red}{+} & \varnothing & \textcolor{blue}{+} & [0 \! \downarrow ] &
\textcolor{red}{+}\tfrac 23 & + & \textcolor{blue}{+} & ( * \! \downarrow ] &
\textcolor{red}{+}\tfrac 13 & - & \textcolor{blue}{+} & [0 \! \downarrow )
\\
&&&&
\textcolor{red}{+} & \varnothing & \textcolor{blue}{-} & [0 \! \uparrow ] &
\textcolor{red}{+}\tfrac 23 & + & \textcolor{blue}{-} & ( * \! \uparrow ] &
\textcolor{red}{+}\tfrac 13 & - & \textcolor{blue}{-} & [0 \! \uparrow )
\\
&&&& \textcolor{red}{-} & \varnothing & \textcolor{blue}{+} & [\uparrow \! 0] &
\textcolor{red}{-}\tfrac 23 & - & \textcolor{blue}{+} & [ \uparrow \! * ) &
\textcolor{red}{-}\frac 13 & + & \textcolor{blue}{+} & (\uparrow \! 0 ]
\\
\varnothing & \varnothing & \textcolor{blue}{-} & [\downarrow \! * ] & 
\textcolor{red}{-} & \varnothing & \textcolor{blue}{-} & [\downarrow \! 0 ] &
\textcolor{red}{-}\tfrac 23 & - & \textcolor{blue}{-} & [ \downarrow \! * ) &
\textcolor{red}{-}\tfrac 13 & + & \textcolor{blue}{-} & (\downarrow \! 0 ] 
\end{array}$
\end{center}
\end{table} 

(\textsc{b}) \textit{Color charge from $o_i = o_{0jk}$:}
In our model, we define color charge $o_i$, $i \in \{1,2,3\}$, to be the spatial version of electric charge $o_0$; see Table \ref{c s table}.
In particular, there are three color charges because space is three dimensional.
Choosing the unphysical orientation of $\tilde{M}$ to be $o_{\tilde{M}} = 1$, we have $o_i = o_{0jk}$.
Thus, the color charge $o_i$ arises from the vanishing of a worldvolume with unit tangent vectors $e_0, e_j, e_k$.
In this article \textit{we will consider the simplest compact surface without boundary, namely the sphere $S^2$, as the spatial slice of a worldvolume with color charge.}
All of the points on such a sphere $S^2 \subset \tilde{M}$ are identified, and therefore the sphere is a single point in spacetime $M$.
We call this geometric object a \textit{pointal sphere}.

Since the points on a pointal sphere $S^2 \subset \tilde{M}$ are identified in spacetime $M$, the (spatial) interior of $S^2$ is topologically a $3$-sphere, $S^3$, in $M$.
(Quite curiously, the isospin group $\operatorname{SU}(2)$ is also topologically $S^3$.
This suggests that our model may be able to reproduce the topological features of the skrymion model of nuclei directly from the geometry of spacetime itself, without the need for solitons in a pion field.) 

All three color charges are present on a pointal sphere $S^2$ since there are points on $S^2 \subset \tilde{M}$ with vanishing tangent subspaces 
\begin{equation*}
e_0 \wedge e_2 \wedge e_3, \ \ \ \ \ e_0 \wedge e_3 \wedge e_1, \ \ \ \ \ e_0 \wedge e_1 \wedge e_2,
\end{equation*} 
and thus with free orientations $o_{023}, o_{031}, o_{012} \in \{ \pm 1 \}$.
Furthermore, since the sphere is oriented and connected, we have $o_{023} = o_{031} = o_{012}$.
Whence, $o_1 = o_2 = o_3$.

A pointal sphere may be viewed as a pointon that has been blown up (like blowing up a balloon).
Let $s$ be the spin vector of a pointon, say $s = o_{jk}e_i$, and blow the pointon up into a pointal sphere $S^2$.
Since each point on $S^2$ in $\tilde{M}$ is identified in $M$, there is a copy of $s$ at each point on $S^2$.
However, $s$ is only unaffected by the blowing up at the point $p_i \in S^2 \subset \tilde{M}$ pierced by $o_ie_i$.
Indeed, $s$ is normal to $S^2$ at $p_i$.
The copies of $s$ along the equator orthogonal to $e_i$ lie in tangent spaces of $S^2$ and thus vanish completely in $M$.
The color $i$ of $S^2$ is therefore distinguished by $s$, and so we say $S^2$ has color charge $o_i = o_{0jk}$.

The copies of $s$ at the other two colors $j,k$---at the points $p_j,p_k \in S^2 \subset \tilde{M}$ pierced by $o_je_j$ and $o_ke_k$---lie along the equator and so vanish.
These color charges have no spatial component in the direction of $s$.
Consequently, they remain purely timelike, that is, \textit{they are electric charges}. 
The sum of the three spin vectors at $p_i,p_j,p_k$ coalesce to $s$ as we blow $S^2$ back down to a pointon. 
Thus, since a pointon has electric charge $\pm 1$, we may take the electric charge of the two dormant colors to be $o_i \tfrac 13$.
Pointal spheres therefore have electric charge $\pm \tfrac 23$.
In our model, pointal spheres represent up and anti-up quarks.

A pointal sphere with color $i$ therefore has four possible states,
\begin{equation*}
o_i = o_{0jk} \in \{ \pm 1 \} \ \ \ \ \text{ and } \ \ \ \ o_{jk} \in \{ +1 \! = \, \uparrow, \, -1 \! = \, \downarrow \},
\end{equation*}
where $o_{jk}$ is its spin in the $e_i$ direction. 
These states are given in Table \ref{first generation table}.

\begin{table}
\caption{A combinatorial derivation of the standard model particles from the free Dirac Lagrangian, with precisely one new massive spin-$2$ boson $x$ \cite{B2}. 
Subscripts $a,b \in \{\uparrow, \downarrow \}$, with $a \not = b$, denote spin states.}
\label{SM}
\begin{center}
    \begin{tabular}{rrrr}
\hline
 & $\gamma_{a} = [ a b, 00]$ & $Z_{a} = [ a a, 00]$ & $x_{ab} = [a 0, 0 b]$\\
 & $Z_0 = [ ab, ba]$ & $H = [aa, bb]$ & $x_0 = [ \downarrow \! {*}, {*} \! \downarrow]$ \\
\hdashline
 $\nu_e = [00, {*} \! \downarrow]$ & $\nu_{\mu} = [\downarrow \uparrow, {*} \! \downarrow]$ & $\nu_{\tau} = [\uparrow \uparrow, {*} \! \downarrow]$ & $\text{\footnotesize{$W_{0}^+$}} = [{*} \! \downarrow, 0 \! \uparrow]$\\
$\bar{\nu}_e = [00, \downarrow \! {*}]$ & $\bar{\nu}_{\mu} = [\uparrow \downarrow, \downarrow \! {*}]$ & $\bar{\nu}_{\tau} = [\uparrow \uparrow, \downarrow \! {*}]$ & $\text{\footnotesize{$W_{0}^-$}} = [\downarrow \! {*}, \uparrow \! 0]$\\
\hdashline
$e_{a} = [00, a 0]$ & $\mu_{a} = [b a, a 0]$ & $\tau_{a} = [bb, a0]$ & $\text{\footnotesize{$W^-_{a}$}} = [{*} \! \downarrow, a 0]$ \\
$\bar{e}_{a} = [00, 0b ]$ & $\bar{\mu}_{a} = [ba, 0b ]$ & $\bar{\tau}_{a} = [aa, 0b]$ & $\text{\footnotesize{$W_{a}^+$}} = [\downarrow \! {*}, 0b]$ \\
\hline
$u_{a} = (  00, {*} b]$ & $c_{a} = (ba, {*} b]$ 
& $t_{a} = (aa, {*} b]$ & \\
$\bar{u}_{a} = [00, a {*}  ) $ & $\bar{c}_{a} = [b a, a {*}  ) $ & $\bar{t}_{a} = [bb, a {*}  ) $ &
\\
\hdashline
$d_{a} = (  00, a 0 ]$ & $s_{a} = (ba, a0 ]$ & $b_{a} = (bb, a0 ]$ & \\
$\bar{d}_{a} = [00, 0 b )$ & $\bar{s}_{a} = [ab, 0 b)$ & $\bar{b}_{a} = [aa, 0b)$ &\\
\hline
    \end{tabular}
  \end{center}
\end{table}

(\textsc{c}) \textit{A consequence of the identification of electric charge with $o_0$:}
By identifying the timelike orientation $o_0$ with electric charge, \textit{spinor chirality is converted from handedness to electric charge}, where the chiral spinors
\begin{equation*}
 \psi^+ = \tfrac 12 (1 + \gamma^5)\psi \ \ \ \ \text{ and } \ \ \ \ \psi^- = \tfrac 12 (1- \gamma^5)\psi
\end{equation*}
have positive and negative charges, respectively \cite[Proposition 2.2]{B2}. 
Consequently, the mass terms of the chiral decomposition of the Dirac Lagrangian (\ref{modified Dirac}) become couplings between spinors of opposite electric charge,
\begin{equation} \label{chiral decomp}
\mathcal{L} = \bar{\psi}(i \slashed \partial - \omega) \psi = \bar{\psi}^+ i \slashed \partial \psi^+ + \bar{\psi}^- i \slashed \partial \psi^- + \omega \bar{\psi}^+ \psi^- + \omega \bar{\psi}^- \psi^+.
\end{equation}
Pairs of spinors that do not violate the Pauli exclusion principle are then bound together by the mass term $\omega \bar{\psi}\psi$ of $\mathcal{L}$; these pairs are called \textit{geoms} for `geometric atoms'. 
A geom consisting of an even (resp.\ odd) number of spinors is a boson (resp.\ fermion), and the electric charge of a geom is the sum of its constituent charges.

The set of all possible geoms yields precisely the standard model particles, minus the gluons and plus one new massive spin-$2$ boson; see Table \ref{SM}.
Notably, hidden within $\mathcal{L}$ is the correct spin, electric charge, and color charge of each fermion (in exactly three generations), electroweak boson, and Higgs boson \cite[Theorem 7.1]{B2}.
Gluons do not arise as freely propagating particles in the model.
The first generation fermions are the geoms with precisely one nonzero spinor, given in Table \ref{first generation table}.

In the table, the spinors $(a0], [0a)$ have electric charge $\mp \tfrac 13 = \pm \tfrac 23 \mp 1$. 
To realize these spinors geometrically, we need an object with the corresponding vanishing subspaces. 
The simplest such object is formed by the union of a pointal sphere with a pointon of opposite electric charge.
We call this object a \textit{pointal union}; in our model it represents the down and anti-down quarks.

\section{Spacetime metrics for pointons, pointal spheres, and their unions} \label{first section}

Consider a single timelike pointon at rest, and let $r_0$ be the radius of its spinor particle. 
We want to determine a metric $ds^2 = g_{\mu \nu} dx^{\mu} dx^{\nu}$ for the resulting spacetime.
Since time is stationary along the pointon's worldline, and assuming asymptotic flatness, the coordinate time vector $\partial_t := \partial/\partial t$ satisfies
\begin{equation*}
\partial_t |_{r = 0} = 0 \ \ \ \ \ \text{ and } \ \ \ \ \ \partial_t |_{r \to \infty} = e_t,
\end{equation*}
where $r$ is the coordinate distance to the pointon.
In contrast, spacelike vectors orthogonal to $\partial_t$ are unaffected by the presence of the pointon.
We may thus take the metric to be of the form
\begin{equation*}
ds^2 = -A(r) dt^2 + dr^2 + r^2 d \Omega^2,
\end{equation*}
where the pointon sits at the origin $r = 0$ in its rest frame, and
\begin{equation*}
\partial_t = A(r)^{1/2} e_t.
\end{equation*}

Recall that the $-g_{00}$ component of the Schwarzschild metric $ds^2$, namely $-g_{00} = 1 - r_s/r$, vanishes at the Schwarzschild radius $r = r_s$.
Thus, translating $r_s$ to the origin $r = 0$,
\begin{equation*}
-g_{00} = 1 - \frac{r_s}{r} \ \ \ \longrightarrow \ \ \ 1 - \frac{r_s}{r + r_s} = \frac{r}{r + r_s},
\end{equation*}
leads us to consider the ansatz
\begin{equation} \label{first ansatz}
A(r) = \left(1 - \frac{r_0}{r+r_0} \right)^{\! m} = \frac{r^m}{(r+r_0)^m},
\end{equation}
where the exponent $m$ determines how quickly $\partial_t$ approaches the unit vector $e_t$ as $r \to \infty$.

Similarly, we would like to determine a metric $ds^2$ for a single pointal sphere of radius $\eta \leq r_0$ at rest, as well as the union of a pointon and pointal sphere.
For brevity, we call their union a \textit{pointal union}. 
Such metrics are of the form
\begin{equation} \label{sphere metric}
ds^2 = -A(r)dt^2 + dr^2 + B(r)r^2 d \Omega^2,
\end{equation}
where $A(r)$, $B(r)$ are continuous functions satisfying 
\begin{equation} \label{constraints for metric}
A(\eta) = B(\eta) = 0, \ \ \ \ \ \ \lim_{r \to \infty} A(r) = \lim_{r \to \infty} B(r) = 1,
\end{equation}
and, for a pointal union, $A(0) = 0$.

The existence of a spinor particle uniquely determines the exponent $m$ in (\ref{first ansatz}) and the functions $A(r)$, $B(r)$: 

\begin{Proposition} \label{sphere metric theorem}
Consider a spacetime $M$ with a single spinor particle $\nu_e$, resp.\ pointon $e$, pointal sphere $u$, or pointal union $d$, centered at the origin $r = 0$ in its inertial frame. 
Let $r_0 > 0$ be the coordinate radius of the spinor particle's trajectory, and let $\eta \in (-r_0, r_0)$.
Then a metric for $M$ is
\begin{equation*}
\begin{array}{lcl}
ds^2 = -\left( \frac{r - \eta}{r + r_0} \right)^{\! 4} dt^2 + dr^2 + \left( \frac{r - \eta}{r + r_0} \right)^{\! 8\eta/(r_0 + \eta)} r^2 d\Omega^2 & \longleftrightarrow & \left\{ \begin{array}{ll} \nu_e & \text{if } \eta \in (-r_0,0) \\ e & \text{if } \eta = 0 \\ u & \text{if } \eta \in (0, r_0) \end{array} \right. 
\\[+4pt]
ds^2 = - \frac{r^2(r - \eta)^2}{(r + r_0)^4} dt^2 + dr^2 + \left( \frac{r - \eta}{r + r_0} \right)^{\! 4\eta/(r_0 + \eta)} r^2 d\Omega^2 & \longleftrightarrow & \ \ \ \, d \ \ \, \text{ if } \eta \in (0, r_0)
\end{array}
\end{equation*}
In the cases of a pointal sphere and pointal union, $\eta$ is the radius of the sphere.
\end{Proposition}

Observe that the metrics are degenerate for $\eta \in [0, r_0)$ and nondegenerate for $\eta \in (-r_0,0)$.
Recall that charge in our model arises from the vanishing of tangent subspaces, that is, from metric degeneracy.
Thus, for the charged configurations $e, u, d$ we have $\eta \in [0,r_0)$, and for the neutral configuration $\nu_e$ we have $\eta \in (-r_0,0)$. 

\begin{proof}
The metrics for a pointon, pointal sphere, and their union are of the form (\ref{sphere metric}).
By varying the action $\int \! ds$ with respect to $r$ we obtain the geodesic equation
\begin{equation} \label{kf}
0 = \frac{d^2r}{dt^2} + \frac{A'}{2} - r(B + r B'/2) \left( \frac{d \theta}{d t} \right)^{\! 2} - r \sin^2 \theta ( B + r B'/2) \left( \frac{d \phi}{dt} \right)^{\! 2}.
\end{equation}
For a circular geodesic centered at the origin, we have $dr = 0$.
Furthermore, it suffices to suppose $\theta = \pi/2$, whence $d \theta = 0$.
Thus, (\ref{kf}) reduces to 
\begin{equation} \label{part 1}
r\left( \frac{d \phi}{dt} \right)^{\! 2} = \frac{A'}{2B + rB'}.
\end{equation}

Since the geodesic is null, we also have $0  = ds^2 = -Adt^2 + Br^2 d \phi^2$. 
Consequently,  
\begin{equation} \label{part 2}
r\left( \frac{d \phi}{dt} \right)^{\! 2} = \frac{A}{rB}.
\end{equation}
The relations (\ref{part 1}) and (\ref{part 2}) then together imply
\begin{equation} \label{equal}
\frac{2}{r} = \frac{A'(r)}{A(r)} - \frac{B'(r)}{B(r)}.
\end{equation}

\textit{(i) Pointon metric.}
First consider our ansatz (\ref{first ansatz}) for a single pointon.
To determine $m$, set $r = r_0$, $\eta = 0$, and $B = 1$ in (\ref{equal}); then
\begin{equation*}
\frac{2}{r_0} = \frac{A'(r_0)}{A(r_0)} = \frac{m}{2r_0}.
\end{equation*}
Thus, $m = 4$. 

\textit{(ii) Pointal sphere metric and spinor particle metric.}
Let $A(r)$, $B(r)$ be the metric coefficients in (\ref{sphere metric}) for a pointal sphere or spinor particle.
A pointal sphere becomes a pointon when its radius $\eta$ goes to zero.
Consequently, their metrics must coincide for $\eta = 0$:
\begin{equation*}
A(r)|_{\eta = 0} = \frac{r^4}{(r+r_0)^4}
\ \ \ \ \ \text{ and } \ \ \ \ \ 
B(r)|_{\eta = 0} = 1.
\end{equation*}
We therefore consider the ansatz
\begin{equation} \label{A and B}
A(r) = \left( \frac{r - \eta}{r + r_0} \right)^{\! 4} \ \ \ \ \ 
\text{ and } \ \ \ \ \ 
B(r) = \left( \frac{r - \eta}{r + r_0} \right)^{\! n},
\end{equation}
where $n = n(\eta)$ is a function of $\eta$ satisfying $n(0) = 0$.\footnote{In Remark \ref{bad ansatz} below we consider a similar ansatz that satisfies the constraints, but show that it does not allow for a pointal union to exist.}

To determine $n$, set $r = r_0$ in (\ref{equal}):
\begin{equation*}
\frac{2}{r_0} = \frac{A'(r_0)}{A(r_0)} - \frac{B'(r_0)}{B(r_0)} = (4 - n) \frac{r_0 + \eta}{2r_0(r_0 - \eta)}.
\end{equation*}
Solving for $n$ we find
\begin{equation*}
n = 8 \frac{\eta}{r_0+\eta}.
\end{equation*}

Note that $n$ is negative for $\eta \in (-r_0,0)$, but this does not produce a singularity in the metric since in this case the metric is nondegenerate for all $r \geq 0$ (other than at the coordinate singularity of the spherical metric $r^2 d\Omega^2$).
Nevertheless, the metric allows for a circular null geodesic with radius $r_0$.
It therefore supports a spinor particle, even though there is no pointon or pointal sphere present.

\textit{(iii) Pointal union metric.}
Finally, let $A(r)$, $B(r)$ be the metric coefficients in (\ref{sphere metric}) for a pointal union.

In terms of force, the Newtonian potential $\Phi$ corresponding to a metric $g_{ab}$ is
\begin{equation} \label{potential}
\Phi = \operatorname{ln}((-g_{tt})^{1/2}).
\end{equation}
Indeed, $\Phi$ satisfies $\nabla_{e_t}e_t = \vec{\nabla} \Phi$, where $\vec{\nabla}\Phi$ is the spatial gradient on the hypersurface orthogonal to $\partial t$ (e.g., \cite[Chapter 11]{F1}).
Furthermore, forces are additive, and thus so too are potentials.

Let $\Phi_1, \Phi_2$ be the respective gravitational potentials of a pointon and pointal sphere, and let $g^1_{ab}, g^2_{ab}$ be their spacetime metrics.
By the additivity of potentials, then, the potential for their union is either the sum or the average of $\Phi_1$ and $\Phi_2$,
\begin{equation*} \label{+}
\Phi = \Phi_1 + \Phi_2 \ \ \ \ \text{ or } \ \ \ \ \Phi = \tfrac 12 (\Phi_1 + \Phi_2).
\end{equation*}
By (\ref{potential}), the time component of the metric $g_{ab}$ of their union is then
\begin{equation*}
g_{tt}  = -e^{2 \Phi} = \left\{ \begin{array}{ll} 
-e^{2(\Phi_1 + \Phi_2)} = -e^{2\Phi_1}e^{2\Phi_2} = - g^1_{tt}g^2_{tt} & \text{ if } \Phi = \Phi_1 + \Phi_2\\
-e^{\Phi_1 + \Phi_2} = -e^{\Phi_1}e^{\Phi_2} = - \big( g^1_{tt}g^2_{tt} \big)^{1/2} & \text{ if } \Phi = \tfrac 12(\Phi_1 + \Phi_2)
\end{array} \right.
\end{equation*}
We therefore consider the ansatz
\begin{equation*} \label{A and B2}
A(r) = \left( \frac{r(r - \eta)}{(r + r_0)^2} \right)^{\! m} \ \ \ \ \ 
\text{ and } \ \ \ \ \ 
B(r) = \left( \frac{r - \eta}{r + r_0} \right)^{\! n}.
\end{equation*}
Observe that if the potentials $\Phi_1$, $\Phi_2$ are added, then $m = 4$, whereas if the potentials are averaged, then $m = 2$.

To determine $m$ and $n$, again set $r = r_0$ in (\ref{equal}):
\begin{equation*}
\frac{2}{r_0} = \frac{A'(r_0)}{A(r_0)} - \frac{B'(r_0)}{B(r_0)} = \frac{m}{r_0 - \eta} - \frac{n(r_0 + \eta)}{2r_0(r_0 - \eta)}.
\end{equation*}
Solving for $n$ we find
\begin{equation*}
n = \frac{2}{r_0 + \eta}\big( r_0(m-2) + 2\eta \big).
\end{equation*}
But $B(r)|_{\eta = 0} = 1$.
Whence $m = 2$, in agreement with the averaging of $\Phi_1$ and $\Phi_2$.
Consequently,
\begin{equation*}
n = 4 \frac{\eta}{r_0+\eta}.
\end{equation*}

We can also determine that $m = 2$ using the spinor structure of a pointal union:
a pointal union corresponds to a \textit{single} spinor. 
In particular, a pointal union becomes a single pointon when its radius $\eta$ goes to zero. 
Indeed, if a pointal union became two pointons when $\eta = 0$, then it would be the bound state of two fermions, and thus a boson, contrary to the fact that a pointal union is a fermion. 
Consequently, the metric of a pointal union must reduce to that of a single pointon when $\eta = 0$.
It follows that $m$ must be $2$. 
\end{proof}

The radius $r_0$ of the spinor particle is analogous to the photon sphere radius $r = 3r_s/2$ of the Schwarzschild metric.

\begin{Remark} \label{bad ansatz} \rm{
Similar to our pointal sphere ansatz (\ref{A and B}), the ansatz
\begin{equation} \label{A and B-2}
A(r) = \left( \frac{r - \eta}{r - \eta + r_0} \right)^{\! m} \ \ \ \ \ 
\text{ and } \ \ \ \ \ 
B(r) = \left( \frac{r - \eta}{r - \eta + r_0} \right)^{\! n},
\end{equation}
with $\lim_{\eta \to 0}n(\eta) = 0$, also satisfies the necessary constraints (\ref{constraints for metric}) for a pointal sphere.
However, this ansatz does not allow pointal unions to exist.
Indeed, following (\textit{ii}) in the proof of Proposition \ref{sphere metric theorem}, the ansatz (\ref{A and B-2}) yields the metric for a pointal sphere,
\begin{equation} \label{wrong pointal sphere metric}
ds^2 = -\left( \frac{r - \eta}{r - \eta + r_0} \right)^{\! 4} dt^2 + dr^2 + \left( \frac{r - \eta}{r - \eta +r_0} \right)^{\! 8 \eta/r_0}r^2 d \Omega^2.
\end{equation}

To determine the pointal union metric from (\ref{wrong pointal sphere metric}), we proceed as in (\textit{iii}) in the proof.
Since a pointal union reduces to a single pointon when the radius $\eta$ vanishes, we have $\Phi = \tfrac 12(\Phi_1 + \Phi_2)$.
Whence, (\ref{wrong pointal sphere metric}) yields the pointal union metric
\begin{equation*} \label{wrong pointal union metric}
ds^2 = -\left( \frac{r(r-\eta)}{(r+r_0)(r-\eta+r_0)}\right)^{\! 2} dt^2 + dr^2 + \left( \frac{r - \eta}{r - \eta +r_0} \right)^{\! -2 + \frac{\eta}{r_0}\left( 5 + \frac{\eta}{r_0} \right)}r^2 d \Omega^2.
\end{equation*}
But then $B(r) \not = 1$ at $\eta = 0$, a contradiction. 
}\end{Remark}

\begin{Remark} \rm{
The \textit{coordinate} centripetal acceleration $r(d\phi/dt)^2$ of the spinor particle appeared in (\ref{part 1}) and (\ref{part 2}).
Its angular frequency $\omega = d \phi/d \tau$ defines the pointon's rest energy $E_0$ by $E_0 = \hbar \omega = \hbar u/r_0 = mcu$, given in (\ref{E0 =}).
Observe that
\begin{equation*}
\omega = \frac{d \phi}{d\tau}(r_0) = \frac{d \phi}{dt} \frac{dt}{d \tau}(r_0) = A(r_0)^{-1/2} \frac{d \phi}{dt} (r_0),
\end{equation*}
where $A(r_0)^{1/2} = 1/4$.
Then, using either (\ref{part 1}) or (\ref{part 2}) with $B = 1$ we obtain\footnote{The factor of $c^2$ was suppressed in (\ref{sphere metric}).}
\begin{equation*}
\omega = A(r_0)^{-1/2} \frac{d \phi}{dt}(r_0) = \frac{c}{r_0}. 
\end{equation*}
Consequently, the relation $\omega r_0 = c$ holds, as it should for an on-shell pointon, that is, a pointon for which $u = c$.
}\end{Remark}

\begin{Proposition} \label{prop2}
For fixed $r_0 > 0$, the radius $\eta \in (0, r_0)$ of a pointal sphere, resp.\ pointal union, can only take the values
\begin{equation} \label{stable}
\eta \in \{ r_0/7, \ r_0/3, \ 3r_0/5 \}, \ \ \ \text{resp.} \ \eta = r_0/3.
\end{equation}
In contrast, if $\eta$ is negative, then $\eta$ can take any value in $(-r_0,0)$.
\end{Proposition}

\begin{proof}
The metric coefficients in a spacetime metric are real valued; in particular, the function $B(r)$ in the pointal sphere metric must be real valued.
Consequently, the exponent $n = 8 \eta/(r_0 + \eta)$ must be an integer since $(r - \eta)/(r+r_0) < 0$ for $r < \eta$.
Thus, since $\eta \in (0, r_0)$, we have $n \in \{1,2,3\}$.
Furthermore, $\eta = nr_0/(8-n)$, yielding (\ref{stable}).

Similarly, the function $B(r)$ in the pointal union metric must be real valued.
Thus, $n = 4 \eta/(r_0 + \eta)$ must be integer.
But in this case, $\eta \in (0, r_0)$ implies that $\eta$ can only take the value $\eta = r_0/3$.

Finally, if $\eta$ is negative, then $(r - \eta)/(r+r_0) > 0$ for all $r \geq 0$, and so there are no constraints on $n$.
\end{proof}

\newpage
\begin{Theorem} \label{sphere metric theorem2}
Consider a spacetime $M$ with a single spinor particle $\nu_e$, resp.\ pointon $e$, pointal sphere $u$, or pointal union $d$, centered at the origin $r = 0$ in its inertial frame. 
Let $r_0 > 0$ be the coordinate radius of the spinor particle's null geodesic trajectory, and let $\eta \in (-r_0, r_0)$.
Then a metric for $M$ is
\begin{equation*}
\begin{array}{lcl}
ds^2 = -\left( \frac{r - \eta}{r + r_0} \right)^{\! 4} dt^2 + dr^2 + \left( \frac{r - \eta}{r + r_0} \right)^{\! 8\eta/(r_0 + \eta)} r^2 d\Omega^2 & \longleftrightarrow & \left\{ \begin{array}{ll} \nu_e & \text{if } \eta \in (-r_0,0) \\ e & \text{if } \eta = 0 \\ u 
& \text{if } \eta \in r_0\left\{ \tfrac 17, \tfrac 13, \tfrac 35 \right\} \end{array} \right.
\\[+4pt]
ds^2 = - \frac{r^2(r - r_0/3)^2}{(r + r_0)^4} dt^2 + dr^2 + \frac{r - r_0/3}{r + r_0} \, r^2 d\Omega^2 & \longleftrightarrow & \ \ \ \, d \ \ \, \text{ if } \eta = r_0/3
\end{array}
\end{equation*}
\end{Theorem}

\begin{proof}
Follows from Propositions \ref{sphere metric theorem} and \ref{prop2}.
\end{proof}

\section{A bare quark mass ratio and the bare electron mass} \label{energy density section}

As shown in Section \ref{first section}, spacetime metrics for a single pointon, pointal sphere, and pointal union are of the form
\begin{equation} \label{general metric}
ds^2 = -A(r) dt^2 + dr^2 + B(r) r^2 d\Omega^2, \ \ \ \ \ B(r) = \left( \frac{r- \eta}{r+r_0} \right)^{\! n}.
\end{equation}
In Proposition \ref{prop2} we found that for a fixed spinor radius $r_0$, the down quark (a pointal union) has a unique sphere radius $\eta = r_0/3$, whereas an up quark (a pointal sphere) has precisely three possible sphere radii, $\eta \in \{ r_0/7, \ r_0/3, \ 3r_0/5 \}$.
We denote these up quarks by $u_1, u_2, u_3$ respectively, where the subscript specifies the exponent $n \in \{1, 2, 3\}$ in $B(r)$ in the pointal sphere metric.

Although the up quark $u_1$ and down quark $d$ have different sphere radii, namely $\eta = r_0/7$ and $\eta = r_0/3$, we will nevertheless find that, among the three up quarks, $u_1$ is the closest to $d$ in multiple ways: 
\begin{itemize}
 \item $u_1$ and $d$ have equal exponents $n = 1$ in the metric coefficient $B(r)$ of (\ref{general metric}).
 \item $u_1$ and $d$ generate the closest energy densities by a significant amount (Theorem \ref{energy densities}).
 \item The $u_1$ radius $\eta = r_0/7$ appears in the acceleration $\nabla_{e_t}e_t$ resulting from $d$ (Proposition \ref{accelerations}).
\end{itemize}
We therefore consider $u_1$ to be the up quark that appears in hadrons at low energies.

\begin{Proposition} \label{rho proposition}
Given a metric of the form (\ref{general metric}), the energy density $\rho = \rho(r)$ obtained from the Einstein field equation is given by
\begin{equation} \label{rho's}
\kappa \rho = \kappa T^{tt} = g^{tt}\big( R^{t}_{t} - \tfrac 12 R \big) = -\frac{1}{A} \! \left( \frac 32 \frac{B''}{B} + \frac{3}{r} \frac{B'}{B} - \frac 14 \frac{B'^2}{B^2} - \frac{1}{r^2} \! \left( \frac{1}{B} - 1 \right) \! \right).
\end{equation} 
\end{Proposition}

\begin{proof}
The nonzero Christoffel symbols of (\ref{general metric}) are
\begin{equation} \label{Christoffel2}
\begin{split}
\Gamma^t_{rt} & = \frac{A'}{2A}, \ \ \ \ 
\Gamma^r_{tt} = \frac{A'}{2}, \\
\Gamma^r_{\theta \theta} & = -r(B + rB'/2), \ \ \ \ 
\Gamma^r_{\phi \phi} = -r \sin^2 \theta \left(B + r B'/2 \right), \\
\Gamma^{\theta}_{r \theta} & = \Gamma^{\phi}_{r \phi} = \frac 1r + \frac{B'}{2B}, \ \ \ \ 
\Gamma^{\theta}_{\phi \phi} = -\sin \theta \cos \theta, \ \ \ \
\Gamma^{\phi}_{\theta \phi} = \frac{\cos \theta}{\sin \theta}.
\end{split}
\end{equation}
The diagonal Ricci tensor components and Ricci scalar are thus
\begin{equation*}
\begin{split}
R^t_t & = -\frac{A''}{2A} + \frac{A'^2}{4A^2} - \frac{A'}{rA} - \frac{A'B'}{2AB}\\
R^r_r & = -\frac{A''}{2A} + \frac{A'^2}{4A} - \frac{B''}{B} + \frac{B'^2}{2B^2} - \frac{2B'}{rB}\\
R^{\theta}_{\theta} = R^{\phi}_{\phi} & = -\frac{1}{r^2} - \frac{2B'}{rB} - \frac{B''}{B} + \frac{1}{r^2B}  - \frac{A'}{2rA} - \frac{A'B'}{4rAB}\\
R = R^{\mu}_{\mu} & = - \! \left( \frac{A''}{A} + 3 \frac{B''}{B}\right) - \frac{2}{r} \! \left( \frac{A'}{A} + 3 \frac{B'}{B} \right) + \frac 12 \! \left( \frac{A'}{A} - \frac{B'}{B} \right)^{\! 2} + \frac{2}{r^2} \! \left( \frac{1}{B} - 1 \right)
\end{split}
\end{equation*}
The energy density $\kappa \rho = \kappa T^{tt} = g^{tt}\big( R^{t}_t - \tfrac 12 R \big)$ in the spatial hypersurface orthogonal to $e_t$ is therefore readily found to be (\ref{rho's}).
\end{proof}

\begin{Theorem} \label{energy densities}
The energy densities $\rho(r)$ at $r = r_0$ of an electron $e$ (= pointon), up quark $u_1, u_2, u_3$ (= pointal sphere), and down quark $d$ (= pointal union) are
\begin{equation*}
\begin{array}{r|lclclclcl}
& e & \ \ & u_1 & \ \ & u_2 & \ \ & u_3 & \ \ & d \\
\hline
\eta = & 0 && \frac{1}{7}r_0 && \frac{1}{3}r_0 && \frac{3}{5}r_0 && \frac{1}{3}r_0 \\
\rho(r_0) \kappa r_0^2 = & 0 && \frac{9604}{729} \approx 13.174 && 243 && \frac{88}{625} \approx .141 && 27
\end{array}
\end{equation*}
Thus, the ratio of the closest up quark and down quark energy densities at $r_0$ is
\begin{equation*}
\frac{\rho_{u_1}(r_0)}{\rho_d(r_0)} = \frac{9604}{729 \cdot 27} = \frac{9604}{19,683} \approx .48793.
\end{equation*}
Furthermore, the electron energy density is identically zero, $\rho_e(r) = 0$.
\end{Theorem}

\begin{proof}
For a pointon, $B = 1$.
Thus, the electron energy density vanishes, ${\rho_e(r) = 0}$, by Proposition \ref{rho proposition}.

Furthermore, it follows from (\ref{rho's}) and Theorem \ref{sphere metric theorem2} that the pointal sphere and pointal union energy densities are respectively
\begin{equation} \label{ny}
\begin{split}
\kappa \rho_{\text{sp}} & = \frac{8 \eta (r+r_0)^2}{(r- \eta)^6}\left( 3\frac{r_0}{r} \eta -\frac{17}{2} \eta - \frac{3}{2} r_0 \right) + \frac{(r+r_0)^{n+4} - (r-\eta)^n(r+r_0)^4}{r^2(r- \eta)^{n+4}}, \\
\kappa \rho_d & = \frac{4}{3}\frac{(r+r_0)^2r_0^2}{(r-r_0/3)^4r^3}\left(r_0 - \frac{8}{3}r \right) + \frac{(r+r_0)^5 - (r+r_0)^4(r-r_0/3)}{r^4 (r-r_0/3)^3}.
\end{split}
\end{equation} 
For $\rho_{\text{sp}}$, we have $n = 8\eta/(r_0 + \eta) \in (-4,0] \cup \{1,2,3\}$, by Theorem \ref{sphere metric theorem2}. 
The theorem then follows by setting $r = r_0$ in (\ref{ny}).
\end{proof}

\begin{Remark} \rm{
In Theorem \ref{energy densities} we found that the energy density of an electron is identically zero, $\rho_e(r) = 0$.
This fact can also be shown directly:
Let $\xi$ be a timelike unit tangent vector, and let $V^3 \subset M^4$ be the spacelike hypersurface normal to $\xi$. 
If $V^3$ is totally geodesic (that is, every geodesic of $V^3$ in the induced metric is also a geodesic of $M^4$), then the scalar curvature of $V^3$ is proportional to the energy density $\rho = T(\xi, \xi)$ by the relation $R_{V^3} = 2\kappa \rho$ (e.g., \cite[4-16]{F1}).

Let $M^4$ have pointon metric $ds^2 = -A(r)dt^2 + dr^2 + r^2d\Omega^2$ given in Proposition \ref{sphere metric theorem}, and let $\xi = e_t$.
Then $V^3 = \mathbb{R}^3$ with the induced metric $d\ell^2 = dr^2 + r^2d\Omega^2$.
Consequently, $2\kappa \rho_e = R_{V^3} = R_{\mathbb{R}^3} = 0$.
}\end{Remark}

\section{Dressed electron/bare quark mass ratios} \label{dressed electron section}

Recall the pointon metric given in Proposition \ref{sphere metric theorem},
\begin{equation} \label{pointon metric here}
ds^2 = \frac{-r^4}{(r + r_0)^4}dt^2 + dr^2 + r^2 d\Omega^2.
\end{equation}
In the following we show that the pointon metric is not a vacuum solution, $R_{ab} \not = 0$, in contrast to the Schwarzschild metric.

\begin{Theorem} \label{pointon stress-energy theorem}
By the Einstein field equation $R_{ab} - \tfrac 12 Rg_{ab} = \kappa T_{ab}$, the pointon metric is equivalent to a fluid with $4$-velocity $u = \partial_t$, energy density $\rho_e = 0$, and pressure and shear viscosity
\begin{equation} \label{pressure}
p_e = \frac{4}{\kappa (r+r_0)^2} \frac{r_0^2}{r^2}, \ \ \ \ \ (\tensor{\pi}{^{\mu}_{\nu}}) = \frac{1}{\kappa (r+r_0)^2}\operatorname{diag}\left( 0, \frac{r_0}{r}, -\frac{1}{2}\frac{r_0}{r},  -\frac{1}{2}\frac{r_0}{r}\right).
\end{equation}
\end{Theorem}

\begin{proof}
The nonzero Christoffel symbols for the pointon metric (\ref{pointon metric here}) are, from (\ref{Christoffel2}),
\begin{equation*} \label{Christoffel}
\begin{split}
\Gamma^t_{rt} & = \frac{2}{(r+r_0)}\frac{r_0}{r}, \ \ \ \ \ \
\Gamma^r_{tt} = \frac{2r^3r_0}{(r+r_0)^5}, \ \ \ \ \ \
\Gamma^r_{\theta \theta} = -r, \ \ \ \ \ \
\Gamma^r_{\phi \phi} = -r \sin^2 \theta,\\
\Gamma^{\theta}_{r \theta} & = \Gamma^{\phi}_{r \phi} = \frac 1r, \ \ \ \ \ \
\Gamma^{\theta}_{\phi \phi} = -\sin \theta \cos \theta, \ \ \ \ \ \
\Gamma^{\phi}_{\theta \phi} = \cot \theta.
\end{split}
\end{equation*} 
The nonzero Ricci tensor components and Ricci scalar are thus
\begin{equation} \label{Ricci}
\begin{split}
R_{tt} & = \frac{6 r^2 r_0^2}{(r+r_0)^6}, \ \ \ \ \
R_{rr} = \frac{4}{(r+r_0)^2}\left(\frac{r_0}{r} - \frac{r_0^2}{2r^2} \right),\\
R_{\theta \theta} & = \frac{-2r_0}{r + r_0}, \ \ \ \ \
R_{\phi \phi} = \sin^2 \theta R_{\theta \theta}, \ \ \ \ \ R = \frac{-12}{(r + r_0)^2}\frac{r_0^2}{r^2}.
\end{split}
\end{equation}
Consequently, the Einstein tensor components $G_{\mu \nu} := R_{\mu \nu} - \tfrac 12 R g_{\mu \nu}$ are
\begin{equation} \label{G}
\begin{split}
G_{tt} & = 0, \ \ \ \ \
G_{rr} = \frac{4}{(r+r_0)^2}\left( \frac{r_0}{r} + \frac{r_0^2}{r^2} \right), \\
G_{\theta \theta} & = \frac{-4}{(r+r_0)^2}\left(  \tfrac 12 r_0 r - r_0^2 \right), \ \ \ \ \
G_{\phi \phi} = \sin^2 \theta G_{\theta \theta}.
\end{split}
\end{equation}

Since the vacuum equation $R_{ab} = 0$ does not hold around a pointon, its effect is equivalent to a fluid.
The stress-energy tensor for a fluid with 4-velocity $u^a$ is
\begin{equation*}
T^{ab} = \rho u^a u^b + p (g^{ab} + u^a u^b) + (u^a q^b + q^a u^b) + \pi^{ab},
\end{equation*}
where $\rho$  is the energy density, $p$ is the isotropic pressure, $q^{a}$ is the heat flux vector, and $\pi^{ab}$ is the viscous shear tensor. 
The viscous shear tensor $\pi^{ab}$ is symmetric, traceless, and orthogonal to $u^a$ in the sense that $\pi_{ab} u^a = 0$. 
In our case we take $u = e_t$.

Lowering an index $\tensor{T}{^{\mu}_{\nu}} = g_{\nu \lambda} T^{\mu \lambda} = \tfrac{1}{\kappa}g_{\nu \lambda}G^{\mu \lambda}$ and using (\ref{G}), we obtain
\begin{equation*}\label{T}
(\tensor{T}{^{\mu}_{\nu}}) = \frac{4}{\kappa (r+r_0)^2} \operatorname{diag}\left( 0, \left( \frac{r_0}{r} + \frac{r_0^2}{r^2} \right), \left(-\frac 12 \frac{r_0}{r} + \frac{r_0^2}{r^2} \right), \left( -\frac 12 \frac{r_0}{r}+ \frac{r_0^2}{r^2} \right) \right).
\end{equation*}
Furthermore, the pressure term of $\tensor{T}{^{\mu}_{\nu}}$ is of the form $p\left(\tensor{g}{^{\mu}_{\nu}}+(e_t)^{\mu}(e _t)_{\nu} \right) = p\tensor{\delta}{^{\mu}_{\nu}}$.
Thus, the pressure is readily seen to be (\ref{pressure}).
The remaining (spatial) part of $\tensor{T}{^{\mu}_{\nu}}$ is the shear viscosity $\pi_{ab}$ in (\ref{pressure}).
This tensor is clearly traceless, symmetric, and satisfies $\pi_{ab}(\partial_t)^{a} = 0$, as it should. 
\end{proof}

The `potential' $\sqrt{-g_{tt}}$ of a static metric $g_{ab}$ satisfies Levi-Cevita's generalized Poisson equation (e.g., \cite[3-8]{F1})
\begin{equation*} \label{Laplacian relations}
\vec{\nabla}^2 \sqrt{-g_{tt}} = -R^t_t \sqrt{-g_{tt}}.
\end{equation*}
Here, the Laplacian $\vec{\nabla}^2 = \nabla^{\alpha}\nabla_{\alpha}$ is with respect to the spatial hypersurface orthogonal to $\partial_t$, 
\begin{equation*}
\vec{\nabla}^2 f = g^{\alpha \beta} f_{; \alpha \beta} = g^{\alpha \beta} \left( \partial_{\beta} \partial_{\alpha} f - \Gamma^{\gamma}_{\beta \alpha} \partial_{\gamma} f \right).
\end{equation*}
For a perfect fluid ball or a dust cloud, the generalized Poisson equation is related to the energy density $\rho$ and pressure $p$ by (e.g., \cite[Chapter 3]{F1}, \cite[Section 11.1b,c]{F2})
\begin{equation} \label{generalized}
\vec{\nabla}^2 \sqrt{-g_{tt}} = \tfrac{\kappa}{2} ( \rho + 3 p) \sqrt{-g_{tt}}.
\end{equation}
In the following, we show that this relation also holds for pointons.

\begin{Proposition}
For a pointon we have
\begin{equation} \label{kfti}
\vec{\nabla}^2 \sqrt{-g_{tt}} = \tfrac{\kappa}{2} ( \rho_e + 3 p_e) \sqrt{-g_{tt}}.
\end{equation}
\end{Proposition}

\begin{proof}
For a metric of the form $ds^2 = -A(r)dt^2 + dr^2 + B(r)r^2d\Omega^2$, the generalized Poisson equation is given by
\begin{equation} \label{sngu}
\vec{\nabla}^2 \sqrt{-g_{tt}} =  \partial^2_r \sqrt{-g_{tt}} + \left( \frac{2}{r} + \frac{B'}{B} \right) \partial_r \sqrt{-g_{tt}}.
\end{equation}
Indeed, since $g_{ab}$ is diagonal, we have
\begin{equation*}
\vec{\nabla}^2 \sqrt{-g_{tt}} = g^{\alpha \alpha} \left( \partial^2_{\alpha} \sqrt{-g_{tt}} - \Gamma^{\gamma}_{\alpha \alpha} \partial_{\gamma} \sqrt{-g_{tt}} \right).
\end{equation*}
Thus, since $g^{\theta \theta} = (Br^2)^{-1}$, $g^{\phi \phi} = (Br^2 \sin^2 \theta)^{-1}$, the Christoffel symbols (\ref{Christoffel2}) yield 
\begin{equation*}
\begin{split}
\vec{\nabla}^2 \sqrt{-g_{tt}} & = g^{rr} \partial^2_r \sqrt{-g_{tt}} - \left( g^{\theta \theta} \Gamma^r_{\theta \theta} + g^{\phi \phi} \Gamma^r_{\phi \phi} \right) \partial_r \sqrt{-g_{tt}}\\
& = \partial^2_r \sqrt{-g_{tt}} + \left( \frac{2}{r} + \frac{B'}{B} \right) \partial_r \sqrt{-g_{tt}}.
\end{split}
\end{equation*}

It follows from (\ref{sngu}) that the pointon metric (\ref{pointon metric here}) satisfies
\begin{equation*}
\sqrt{-g^{tt}}\vec{\nabla}^2 \sqrt{-g_{tt}} = \frac{6}{(r+r_0)^2}\frac{r_0^2}{r^2} = -R^t_t.
\end{equation*}
The right-hand side is related to the energy density $\rho_e$ and pressure $p_e$ obtained from the Einstein field equation in (\ref{pressure}) by
\begin{equation*}
\rho_e = 0 \ \ \ \ \text{ and } \ \ \ \ p_e = \frac{2}{3\kappa} \frac{6}{(r+r_0)^2}\frac{r_0^2}{r^2}.
\end{equation*}
Therefore the relation (\ref{kfti}) holds.
\end{proof}

Denote by $\tilde{\rho} := \rho + 3p$ the total energy density due to both the energy density $\rho$ and pressure $p$ occurring in the generalized Poisson equation,
\begin{equation*}
\vec{\nabla}^2 \sqrt{-g_{tt}} = -R^t_t \sqrt{-g_{tt}} = \tfrac{\kappa}{2}( \rho + 3p )\sqrt{-g_{tt}} = \tfrac{\kappa}{2} \tilde{\rho} \sqrt{-g_{tt}}.
\end{equation*}
In particular, 
\begin{equation} \label{rho^*}
-R^t_t = \tfrac{\kappa}{2}(\rho + 3p) = \tfrac{\kappa}{2}\tilde{\rho}.
\end{equation}

In Section \ref{energy density section}, we identified the energy density $\rho$ with the bare mass (density) of the particle.
Similarly, we may identify the pressure $p$ with the mass that results from the polarization of the vacuum, that is, from the particle's covering of virtual particles. 
Under these identifications, the total energy density $\tilde{\rho}$ may be identified with the particle's dressed mass through the generalized Poisson equation (\ref{generalized}): 
\begin{center}
\begin{tabular}{|rcl|}
\hline
internal geometry + EFE & & vacuum polarization in QFT\\
\hline
energy density $\rho$ & $\longleftrightarrow$ & bare mass\\
isotropic pressure $p$ & $\longleftrightarrow$ & mass from virtual particle covering\\
total energy density $\tilde{\rho} = \rho + 3p$ & $\longleftrightarrow$ & dressed mass\\
\hline
\end{tabular}
\end{center}

\begin{Remark} \rm{
The notion of a virtual covering (or vacuum polarization) in quantum field theory is not well defined mathematically.
In our framework we have replaced the virtual covering with the well-defined notion of isotropic fluid pressure.
This fluid, in turn, is equivalent to a particular spacetime geometry which is determined by the Einstein field equation.
}\end{Remark}

\begin{Theorem} \label{dressed electron mass theorem}
An electron (pointon) has total energy density
\begin{equation*}
\tilde{\rho}_e(r) = -\tfrac{2}{\kappa}R^t_t = \frac{12}{\kappa(r+r_0)^2}\frac{r_0^2}{r^2}.
\end{equation*}
\end{Theorem}

\begin{proof}
Follows from (\ref{Ricci}) and (\ref{rho^*}).
\end{proof}

\begin{Corollary} \label{main corollary1}
The total energy density of the electron at $r = r_0$ is
\begin{equation*}
\tilde{\rho}_e(r_0) = 3 \kappa^{-1}r_0^{-2}.
\end{equation*}
Thus, the ratios of the dressed electron mass to the bare quark masses are
\begin{equation*}
\frac{\tilde{\rho}_e(r_0)}{\rho_{u_1}(r_0)} = \frac{2187}{9604} \approx .2277 \ \ \ \ \ \ \text{ and } \ \ \ \ \ \ \frac{\tilde{\rho}_e(r_0)}{\rho_d(r_0)} = \frac{1}{9} \approx .1111.
\end{equation*}
Consequently, given the dressed electron mass as $\tilde{m}_e = .511 \operatorname{MeV}$, we derive the bare up and down quark masses:
\begin{equation*}
m_{u_1} \approx 2.2440 \operatorname{MeV} \ \ \ \ \ \ \text{ and } \ \ \ \ \ \ m_d \approx 4.599 \operatorname{MeV}.
\end{equation*}
\end{Corollary}

\begin{proof}
Follows from Theorems \ref{energy densities} and \ref{dressed electron mass theorem}.
\end{proof}

Fix $p \in M$ and let $\{e_0,e_1,e_2,e_3 \}$ be an orthonormal basis for $M_p$. 
Denote by $\varepsilon_i = \left\langle e_i, e_i \right\rangle = \pm 1$ the indicator of $e_i$.
We want to determine the sectional curvatures 
\begin{equation*}
K(e_i \wedge e_j) = \varepsilon_i \varepsilon_j \left\langle R(e_i,e_j)e_j, e_i \right\rangle = \varepsilon_i \varepsilon_j R_{ijij}
\end{equation*}
of the planes $e_i \wedge e_j \subset M_p$, that is, the Gaussian curvatures at $p$ of the $2$-dimensional manifolds formed from the geodesics through $p$ which are tangent to $e_i \wedge e_j$. 

\begin{Proposition}
The sectional curvatures of the pointon metric (\ref{pointon metric here}) are 
\begin{align*} \label{sectional curvatures}
\begin{split}
K(e_t \wedge e_r) & = K(e_t \wedge e_{\theta}) = K(e_t \wedge e_{\phi}) = \frac{-2}{(r+r_0)^2}\frac{r_0^2}{r^2},\\
K(e_{\theta} \wedge e_{\phi}) & = \frac{-4}{(r+r_0)^2}\frac{r_0}{r},\\
K(e_r \wedge e_{\theta}) & = K(e_r \wedge e_{\phi}) = \frac{2}{(r+r_0)^2}\frac{r_0}{r}.
\end{split}
\end{align*}
\end{Proposition}

\begin{proof}
The sectional curvatures readily follow from (\ref{Ricci}), (\ref{G}), and the relations (e.g., \cite[4-6, 4-8]{F1})
\begin{equation*}
R = \sum_{i, j} K(e_i \wedge e_j),  \ \ \ \ \ G_{\mu \mu} = -g_{\mu \mu} \sum_{\substack{i < j;\\ i, j \not = \mu}} K(e_i \wedge e_j).
\end{equation*}
\end{proof}

\section{Gravitational mass} \label{gravitational mass section}

In this section we use the terms `pointon' and `electron' interchangeably. 

Recall that acceleration in the Schwarzschild metric yields Newton's law of gravity,
\begin{equation*}
\nabla_{e_t} e_t = \frac{-r_s}{2r^2} \partial_r = \frac{-Gm}{r^2} \partial_r.
\end{equation*}

\begin{Proposition} \label{accelerations}
The accelerations due to the pointal metrics are
\begin{equation*} \label{acc}
\begin{split}
\text{pointal sphere metric: } \ \ \ \nabla_{e_t}e_t & = \frac{-2(r_0+ \eta)}{(r- \eta)(r+r_0)} \partial_r = \left\{ \begin{array}{ll} \frac{-2(r_0 - |\eta|)}{(r+|\eta|)(r+r_0)} \partial_r & \text{ for $\nu_e$}\\ 
\frac{-2r_0}{r(r+r_0)} \partial_r & \text{ for $e$}\\ 
\frac{-16r_0}{7(r-r_0/7)(r+r_0)} \partial_r & \text{ for $u_1$} 
\end{array} \right.\\
\text{pointal union metric: } \ \ \ \nabla_{e_t}e_t & = \frac{-7(r - r_0/7)}{3(r-r_0/3)(r+r_0)}\frac{r_0}{r} \partial_r \ \ \ \ \ \ \ \ \ \ \ \ \ \ \ \ \ \ \ \ \text{ for $d$}
\end{split}
\end{equation*}
Remarkably, the $u_1$ quark radius $\eta = r_0/7$ appears in the acceleration of the $d$ quark.
\end{Proposition}

\begin{proof}
The acceleration $\nabla_u u = u^{\mu} \nabla_{\mu} u$ of the unit $4$-velocity $u = e_t$ is, by (\ref{Christoffel2}),\footnote{Note that $\left\langle \partial_t, \partial_t \right\rangle = g_{tt}$ and $\left\langle e_t, e_t \right\rangle = -1$ impy $(\partial_t)^t = 1$, $(\partial_t)_t = g_{tt}$, $(e_t)^t = \sqrt{-g^{tt}}$, $(e_t)_t = \sqrt{-g_{tt}}$.}
\begin{equation*}
\nabla_{e_t} e_t = (e_t)^{\mu} \nabla_{\mu} e_t = g^{tt} (e_t)_t \nabla_t \sqrt{-g^{tt}} \partial_t
 = g^{tt} \sqrt{-g_{tt}} \sqrt{-g^{tt}} \nabla_t \partial_t
= g^{tt} \Gamma^r_{tt} \partial_r = -\frac{A'}{2A} \partial_r.
\end{equation*}
\end{proof}

\begin{Corollary}
The acceleration $\nabla_{e_t}e_t$ of each configuration is singular at precisely the degenerate points, that is, at the points that have been `removed' from spacetime.
\end{Corollary}

\begin{Corollary} \label{positive mass}
Each pointal configuration has positive gravitational mass.
\end{Corollary}

\begin{proof}
The acceleration $\nabla_{e_t}e_t$ resulting from the $\nu_e$ configuration, which is not degenerate, is negative in the radial direction $\partial_r$ since $r_0 > |\eta|$.
For each degenerate configuration, the acceleration points towards the degenerate points. 
Therefore the gravitational mass of each configuration is positive.
\end{proof}

The only contribution to the total energy density $\tilde{\rho}_e$ of a pointon is the pressure $p_e$, by Theorem \ref{dressed electron mass theorem}.
Thus, we define the effective gravitational mass $m(r)$ of a pointon at a distance $r$ to be the integral of the pressure (i.e., momentum flux) over the ball $B(r)$ of radius $r$ centered at the pointon, in the spatial hypersurface of its rest frame:
\begin{equation} \label{effective mass}
\begin{split}
m(r)c^2 & = \int_{B(r)} p_e dV \\
& = \int_0^{2 \pi} \int_0^{\pi} \int_0^r \left( \frac{4}{\kappa (x+r_0)^2}\frac{r_0^2}{x^2} \right) x^2\sin \theta \, dx \, d\theta \, d\phi
= \frac{2r_0rc^4}{G(r+r_0)},
\end{split}
\end{equation}
where $\kappa := 8 \pi G/c^4$.

\begin{Theorem}
Newton's law of gravity holds for a pointon,
\begin{equation*}
\nabla_{e_t}e_t(r) = \frac{-G m(r)}{r^2} \partial r.
\end{equation*}
\end{Theorem}

\begin{proof}
We have 
\begin{equation*}
\nabla_{e_t}e_t \stackrel{(\textsc{i})}{=} \frac{-2r_0}{r(r+r_0)} \partial_r \stackrel{(\textsc{ii})}{=} \frac{-Gm(r)}{r^2} \partial_r,
\end{equation*}
where (\textsc{i}) holds by Proposition \ref{accelerations} and (\textsc{ii}) holds by (\ref{effective mass}).
\end{proof}

We expect that Newton's law of gravity also holds for the other pointal metrics.

\section{A constituent quark mass ratio} 

In this section we determine a constituent up/down quark mass ratio.
The constituent mass of a quark is its hadron binding energy.
It is an effective mass in QCD that includes gluons and sea-quarks.

Recall that the interaction radius $r_0$ of a pointon is $r_0 = m_e^{-1}$, defined in (\ref{E0 =}).
In a bound state, $r_0$ also depends on the binding energy $\omega_b$ from the coupling (mass) term in the Dirac Lagrangian,
\begin{equation} \label{r_0 with omega}
r_0^{-1} = \tilde{m}_e + \sum \omega_b.
\end{equation}
We call $\omega_b$ \textit{pointal binding energy}.

\begin{Theorem} \label{constituent}
Suppose that the constituent quarks in a hadron share the same pointal binding energy.
Then, for $r \gg r_0$, the accelerations from $u_1, d$, are
\begin{equation*}
\nabla_{e_t}e_t \sim \left\{ \begin{array}{ll} 
\frac{-(16/7)r_0}{r^2} \partial_r & \text{ for $u_1$} \\
\frac{-(7/3)r_0}{r^2}\partial_r & \text{ for $d$}
\end{array} \right.
\end{equation*}
We therefore obtain a constituent mass ratio for $r \gg r_0$, 
\begin{equation*}
\frac{\tilde{m}_{u_1}}{\tilde{m}_d} = \frac{16/7}{7/3} = \frac{48}{49} \approx .97959.
\end{equation*} 
\end{Theorem}

\begin{proof}
If $u_1, d$ share the same binding energy $\omega_b$, then they share the same interaction radii $r_0$ by (\ref{r_0 with omega}).
The theorem then from Proposition \ref{accelerations}.
\end{proof}

In Proposition \ref{prop2} we found that the radius of the $d$ sphere is $\eta_d = r_0/3$, and the radius of the $u_1$ sphere is $\eta_{u_1} = r_0/7$.
Furthermore, in Proposition \ref{accelerations} we found, quite curiously, that the acceleration $\nabla_{e_t}e_t$ from the $d$ sphere vanishes precisely at the radius of the $u_1$ sphere. 
It is therefore possible for the $u_1$ sphere to sit inside the $d$ sphere and remain at rest. 
This suggests a model in which the constituent quarks of a hadron are nested spheres, centered at the same point and having equal interaction radii $r_0$. 
In particular, if two constituent quarks are of the same flavor, then they would share the same configuration (pointal sphere or pointal union), but possess orthogonal spin vectors corresponding to their color charges. 
We leave this possibility for future investigation.

\section{The energy density of electron neutrinos} \label{neutrino section}

The pointal sphere and pointal union energy densities are given in (\ref{ny}).  
Graphs of $\rho_{\text{sp}}$, with $n = \pm 1$, and $\rho_d$ over the $(r,r_0)$-plane are shown in Figure \ref{cp}.\footnote{The figure was made using GeoGebra\textsuperscript \textregistered, available at \tt{https://www.geogebra.org}.}
On the domain $\{r \geq r_0 > 0\}$, $\rho_d$ is positive. 
Similarly, $\rho_{\text{sp}}$ is \textit{positive} for $n > 0$, \textit{zero} for $n = 0$, and \textit{negative} for $n < 0$:
\begin{equation*}
\operatorname{sgn}\rho_{\text{sp}}|_{r \geq r_0>0} = \operatorname{sgn}n.
\end{equation*}
Thus, for $r \geq r_0$, 
\begin{itemize}
 \item the down quark $d$ and up quarks $u_n$, with $n \in \{1,2,3\}$, generate positive energy densities; 
 \item the electron $e$, with $n = 0$, generates zero energy density; and
 \item the electron neutrino $\nu_e$, with $n \in (-4,0)$, generates a negative energy density.
\end{itemize}
For $r < r_0$, the energy densities may be positive or negative.

\begin{figure}
\includegraphics[width=12cm]{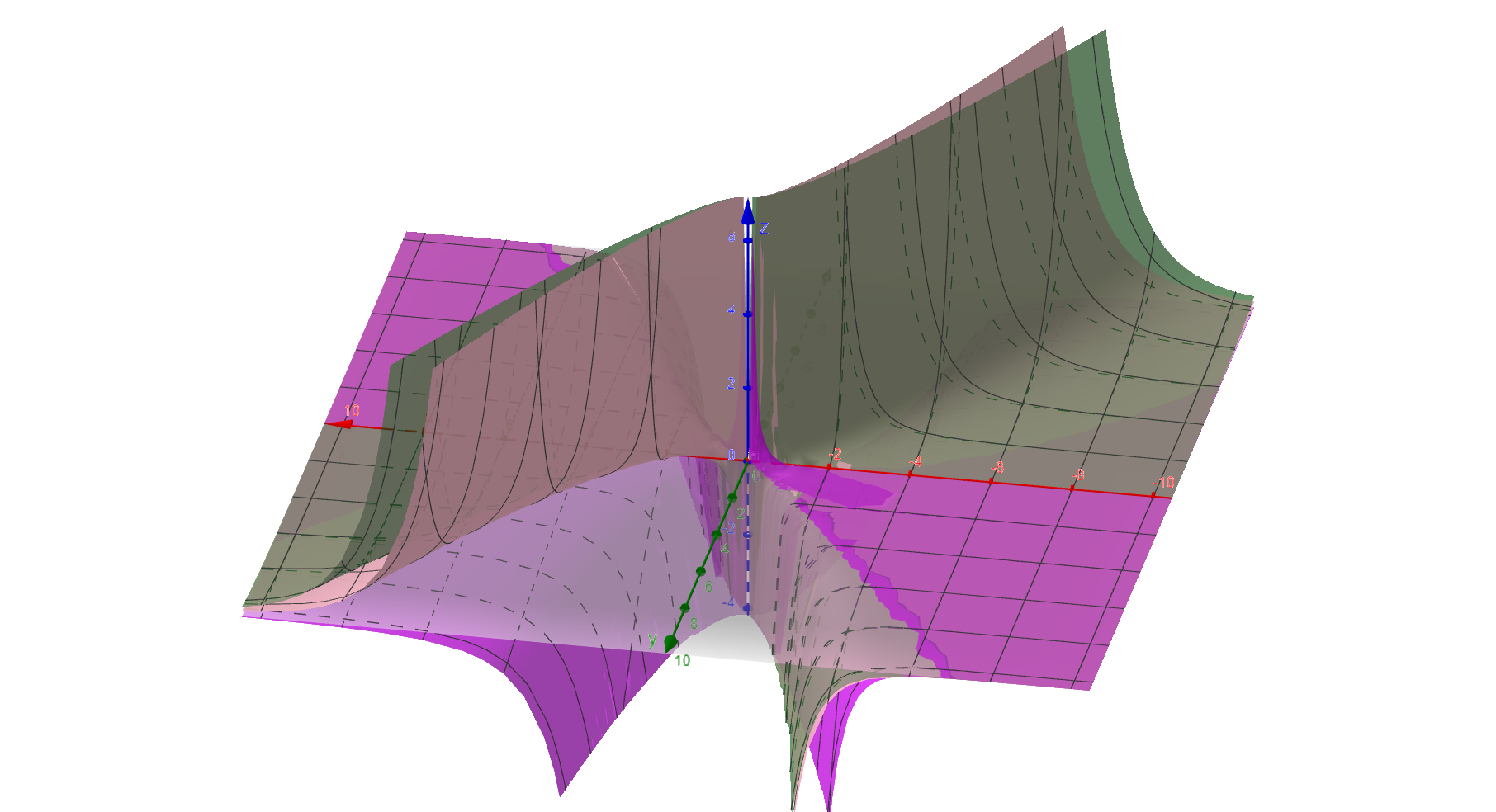}
\caption{Graphs of the energy densities $\rho_{\text{sp}}$, $n = \pm 1$, and $\rho_d$ over the $(r,r_0)$-plane. The graph of $\rho_{\text{sp}}$ with $n = -1$ (drawn in magenta) corresponds to the electron neutrino and lies below the $(r,r_0)$-plane.}
\label{cp}
\end{figure}

Although the electron neutrino generates a negative energy density for $r \geq r_0$, it nevertheless has a positive gravitational mass by Corollary \ref{positive mass}. 
Gravitational energy similarly has a negative energy density.
It is thus possible that an inhomogeneous sea of electron neutrinos would act like additional gravitation.
Since electron neutrinos are weakly interacting, this suggests that they may be a possible dark matter candidate.
It would not be their mass, however, that would yield dark matter, but the negative energy density they generate.
This possibility requires further investigation to determine whether it is a viable model.

\ \\
\textbf{Acknowledgments:}
The author thanks an anonymous referee for their careful reading and helpful comments.
The author was supported by the Austrian Science Fund (FWF) grant P 34854.\\
\\
\textbf{Data Availability Statement:} There is no data associated to the paper.

\bibliographystyle{hep}
\def\cprime{$'$} \def\cprime{$'$}

\end{document}